\documentclass{rAMF2e}
\usepackage[english]{babel}

\theoremstyle{plain}
\newtheorem{thm}{Theorem}[section]
\newtheorem{cor}[thm]{Corollary}
\newtheorem{prop}[thm]{Proposition}

\theoremstyle{definition}
\newtheorem{defn}[thm]{Definition}

\theoremstyle{definition}
\newtheorem{example}[thm]{Example}

\theoremstyle{remark}
\newtheorem{rem}[thm]{Remark}
\newtheorem*{rem*}{Remark}

\usepackage{hhline} 

\begin{document}
\title{Risk-Neutral Pricing and Hedging of In-Play Football Bets}
\date{2015-06-23} 
\author{
	\name{
		Peter Divos$^\ast{}^\dagger$
		\thanks{{\em{Correspondence Address}}:
			Department of Computer Science, University College London, Gower Street, London W1C 6BT, United Kingdom},
		Sebastian del Bano Rollin$^{\ddagger}$,
		Zsolt Bihari$^\S$
		\&
		Tomaso Aste$^\ast{}^\dagger$
	}
	\affil{
	$^\ast$Department of Computer Science, University College London, Gower Street, London W1C 6BT, United Kingdom,  
	$^{\dagger}$Systemic Risk Centre, London School of Economics and Political Sciences, London WC2A 2AE, United Kingdom,
	$^{\ddagger}$School of Mathematical Sciences, Queen Mary University of London, Mile End Road, London E1 4NS, United Kingdom
	$^\S$Department of Finance, Corvinus University of Budapest, 1093 Budapest, Fovam ter 8, Hungary}
}

\maketitle

\begin{abstract}
A risk-neutral valuation framework is developed for pricing and hedging
in-play football bets based on modelling scores by independent Poisson
processes with constant intensities. The Fundamental Theorems of Asset
Pricing are applied to this set-up which enables us to derive novel
arbitrage-free valuation formul\ae\ for contracts currently traded
in the market. We also describe how to calibrate the model to the
market and how trades can be replicated and hedged.
\end{abstract}

\section{Introduction}\label{sec:intro}
In-play football bets are traded live during a football game.
The prices of these bets are driven by the goals scored in the underlying game
in a way such that prices move smoothly between goals and jump to
a new level at times when goals are scored. This is similar to financial markets
where the price of an option changes according to the price changes
of the underlying instrument. We show that the Fundamental Theorems
of Asset Pricing can be applied to the in-play football betting market
and that these bets can be priced in the risk-neutral framework.

Distribution of final scores of football games has been studied by
several authors. In particular, \citet{maher1982modelling} found that an independent
Poisson distribution gives a reasonably accurate description of football
scores and achieved further improvements by applying a bivariate Poisson
distribution. This was further developed by \citet{dixon1997modelling}
who proposed a model in which the final scores of the two teams are
not independent, but the marginal distributions of each team's scores
still follow standard Poisson distributions.

Distribution of in-play goal times has been studied by \citet{dixon1998birth}
who applied a state-dependent Poisson model where the goal intensities
of the teams depend on the current score and time. The model also
accounts for other factors such as home effect and injury time. The
standard Poisson model has been applied by \citet{fitt2005valuation}
to develop analytical valuation formulae for in-play spread bets on
goals and also on corners. A stochastic intensity model has been suggested by
\citet{jottreau2009cir} where the goals are driven by Poisson processes
with intensities that are stochastic,
in particular driven by a Cox-Ingerson-Ross process.
\citet{vecer2009estimating} have shown that
in-play football bets may have additional sensitivities on the top of the
standard Poisson model, for instance sensitivities to red cards.

The Fundamental Theorems of Asset Pricing form the basis of the risk-neutral
framework of financial mathematics and derivative pricing
and have been developed by several authors,
including \citet{cox1976valuation}, \citet{harrison1979martingales},
\citet{harrison1981martingales}, \citet{harrison1983stochastic},
\citet{huang1985information}, \citet{duffie1988security} and \citet{back1991fundamental}.
The first fundamental theorem states that a market is arbitrage free if and only if there
exists a probability measure under which the underlying asset prices
are martingales. The second fundamental theorem states that the market
is complete, (that is, any derivative product of the underlying assets
can be dynamically replicated) if and only if the martingale measure
is unique. 

In this paper we use independent standard time-homogeneous Poisson
processes to model the scores of the two teams. We construct a market of three
underlying assets and show that within this model a unique martingale
measure exists and therefore the market of in-play football bets is
arbitrage-free and complete. Then we demonstrate calibration and replication
performance using market data.

The structure of this paper is the following. Section \ref{sec:inplay} contains a
general overview of in-play football betting and an overview
of the data set. Section \ref{sec:Maths} defines the formal model and contains pricing formulae for
Arrow-Debreu securities among others. In Section \ref{sub:Calibration} we calibrate the model to
historical market quotes of in-play bets and in Section \ref{sec:nextgoal} we use the
same data to show that Next Goal bets are natural hedging instruments
that can be used to build a replicating portfolio to match the values of other bets,
in particular the liquidly traded Match Odds bets.
The Appendix reports analytical pricing formulae for some of the most liquidly traded bets.

\section{In-Play Football Betting}\label{sec:inplay}
In traditional football betting, also known as pre-game or fixed odds
betting, bets are placed before the beginning of the game. In-play
football betting enables bettors to place bets on the outcome of a
game after it started. The main difference is that during in-play
betting, as the game progresses and as the teams score goals, the
chances of certain outcomes jump to new levels and so do the odds
of the bets. Prices move smoothly between goals and jump once a goal is scored.
In-play betting became increasingly popular in recent years.
For instance, \citet{inplaytracker2013}
recently reported that for one particular bookmaker (Unibet) in-play
betting revenues exceeded pre-game betting revenues by 2013Q2 as shown
in Figure \ref{fig:unibet}.

\begin{figure}[t]
\begin{center}
\includegraphics[width=0.55\paperwidth]{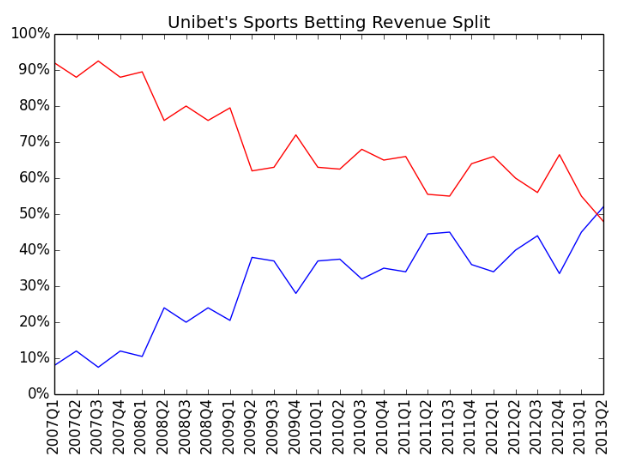}
\end{center}
\caption{\label{fig:unibet}Revenue distribution of one particular bookmaker's
(Unibet) football betting revenues between In-Play and Pre-Game football
betting.}
\end{figure}

There are two main styles of in-play betting: odds betting and spread
betting. In odds betting, the events offered are similar to digital
options in the sense that the bettor wins a certain amount if the
event happens and loses a certain amount otherwise. Typical odds
bets are whether one team wins the game, whether the total number
of goals is above a certain number or whether the next goal is scored
by the home team. In spread betting, the bets offered are such that
the bettor can win or lose an arbitrary amount. A typical example
is a bet called ``total goal minutes'' which pays the bettor the
sum of the minute time of each goal. In this paper we focus on odds
betting, but most of the results can also be applied to spread betting.
A study of spread betting containing analytical pricing formulae for
various spread bets was published by \citet{fitt2005valuation}.

In-play betting offers various types of events such as total goals,
home and away goals, individual player goals, cards, corners, injuries
and other events. This paper focuses on bets related to goal events
only. 

Throughout the paper we refer to the value $X_{t}$ of a bet as the
price at which the bet can be bought or sold at time $t$ assuming
that the bet pays a fixed amount of 1 unit in case it wins and zero
otherwise. This is a convenient notation from a mathematical point
of view, however it is worth noting that different conventions are
used for indicating prices in betting markets. The two most popular
conventions are called fractional odds and decimal odds. Both of these
conventions rely on the assumption that the bettor wagers a fixed
stake when the bet is placed and enjoys a payoff in case the bet wins
or no payoff in case it loses. Fractional odds is the net payoff
of the bet in case the bet wins (that is, payoff minus stake), divided
by the stake. Decimal odds is the total payoff of the bet in case
the bet wins, divided by the stake. Therefore, the value of a bet
$X_{t}$ is always equal to the reciprocal of the decimal odds which
is equal to the reciprocal of fractional odds plus one, formally:
\begin{equation}
X_{t}=\frac{1}{\it{Decimal}_{t}}=\frac{1}{\it{Fractional}_{t}+1},
\end{equation}
where $\it{Decimal}_{t}$ denotes decimal and $\it{Fractional}_{t}$ denotes fractional odds.
Most of the market data we used was originally represented as decimal
odds, but they were converted to bet values using the above formula
for all the figures and for the underlying calculations in this paper.

It is also worth noting that bets can be bought or sold freely during
the game. This includes going short which is referred to as lay betting.
Mathematically this means that the amount held can be a negative number.

In-play bets can be purchased from retail bookmakers
at a price offered by the bookmaker, but can also be traded
on centralized marketplaces where the exchange merely matches orders of participants
trading with each other through a limit order book and keeps a deposit from each party
to cover potential losses.

\subsection{An example game}\label{sec:examplegame}
In order to demonstrate our results we selected the Portugal
vs. Netherlands game from the UEFA Euro 2012 Championship which was
played on the 22nd of June 2012. The reason for selecting this particular
game is that the game had a rather complex unfolding with Netherlands
scoring the first goal, but then Portugal taking the lead in the second
half and finally winning the game. This made the odds jump several times
during the game which makes it a good candidate for demonstrating how
the model performs in an extreme situation. The number of goals as
a function of game time is shown in Figure \ref{fig:Number-of-goals}.

Figures \ref{fig:Match-Odds-values.} and \ref{fig:Over-Under-values.}
show market values of two bet types traded on a betting
market called Betfair: Match Odds and Over-Under. Match Odds contains
three bets: home team winning the game, away team winning the game
and the draw. Over-Under contains bets on the total number of goals
where Under X.5 is a bet that pays off if the total number of goals
is equal or less than X. The dashed lines show the best buy and sell
offers on the market while the continuous lines show the calibrated
model values (see Section \ref{sub:Calibration}).

In case of Match Odds, the value of the bet for Netherlands winning
the game jumped after Netherlands scored the first goal. When the
scores became even after Portugal scored a goal, the value of the
Draw bet jumped up and when Portugal took the lead by scoring the
third goal, the value of the bet for Portugal winning the game jumped
up. Finally, by the end of the game the value of the bet for Portugal
winning the game converged to 1 and the value of the other bets went
to zero.

In case of the Over-Under bets, trading ceased for the Under 0.5 bet
after the first goal when the value of this bet jumped to zero. By the end
of the game, the value of the Under 3.5, 4.5, 5.5, 6.5 and 7.5 bets
reached 1 because the total number of goals was actually 3 and
the values of the Under 0.5, 1.5 and 2.5 bets went to zero.

\begin{figure}[t]
\begin{center}
\includegraphics[width=0.55\paperwidth]{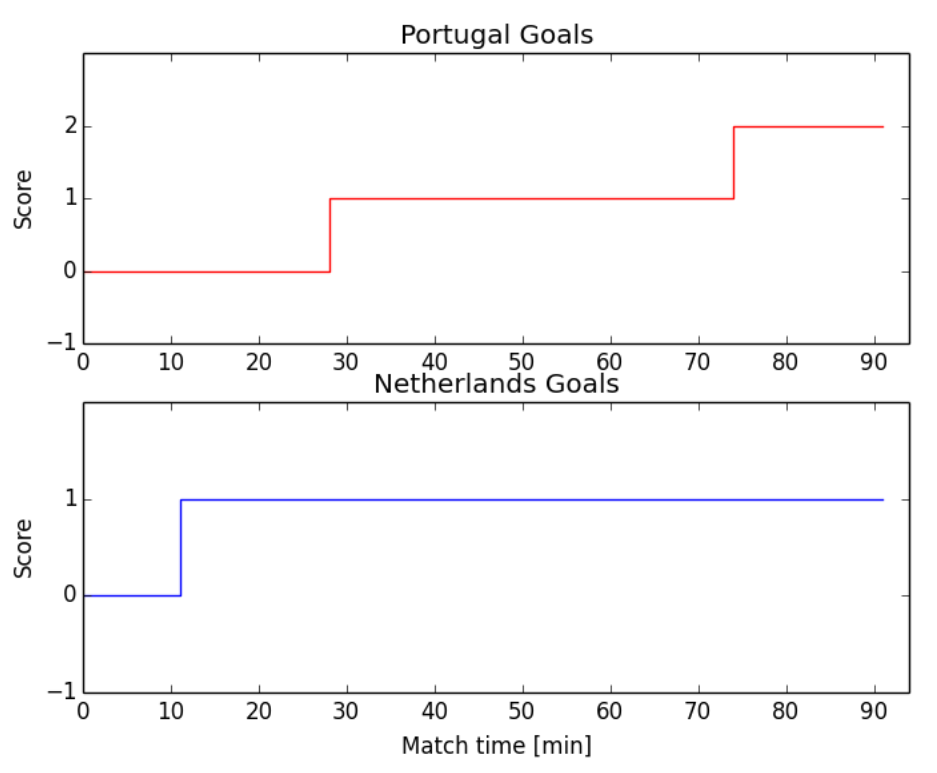}
\end{center}
\caption{\label{fig:Number-of-goals}Scores of the two teams during the Portugal
vs. Netherlands game on the 22nd of June, 2012. The half time result
was 1-1 and the final result was a 2-1 win for Portugal.}
\end{figure}

\begin{figure}[t]
\begin{centering}
\includegraphics[width=0.55\paperwidth]{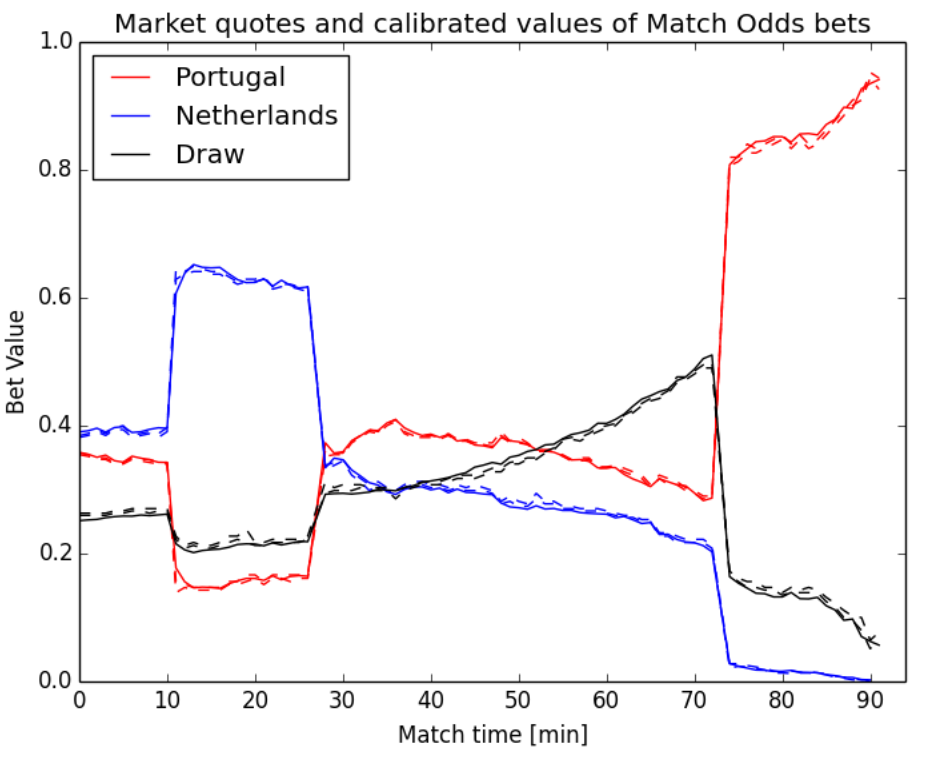}
\par\end{centering}
\protect\caption{\label{fig:Match-Odds-values.}Values of the three Match Odds bets
during the game: Draw (black), Portugal Win (red), Netherlands Win
(blue). Dashed lines represent the best market buy and sell offers
while the continuous lines represent the calibrated model values.
Note that the value of the Netherlands Win bet jumps up after the first
goal because the chance for Netherlands winning the game suddenly
increased. It jumped down for similar reasons when Portugal scored
it's first goal and at the same time the value of the Portugal Win
and Draw bets jumped up. By the end of the game, because Portugal
actually won the game, the value of the Portugal Win bet reached 1
while both other bets became worthless.}
\end{figure}

\begin{figure}[t]
\begin{centering}
\includegraphics[width=0.55\paperwidth]{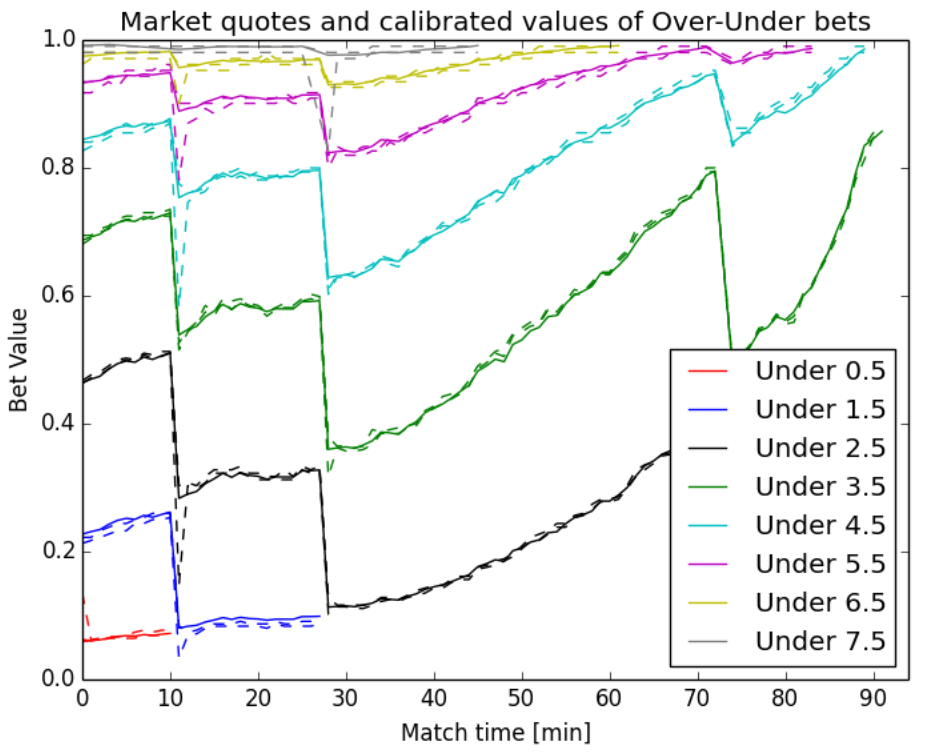}
\par\end{centering}
\protect\caption{\label{fig:Over-Under-values.}Values of Over/Under bets during the
game. Under X.5 is a bet that pays off in case the total number of
goals by the end of the game is below or equal to X. Marked lines
represent the calibrated model prices while the grey bands show the
best market buy and sell offers. Note that after the first goal trading
in the Under 0.5 bet ceased and it became worthless. By the end of
the game when the total number of goals was 3, all the bets up until
Under 2.5 became worthless while the Under 3.5 and higher bets reached
a value of 1.}
\end{figure}

\section{Mathematical framework}\label{sec:Maths}
In this section we present a risk-neutral valuation framework for in-play
football betting. To do so we follow the financial mathematical approach,
in which we start by assuming a probability space, then identify a
market of underlying tradable assets and postulate a model for the
dynamics of these assets. We show that the first and second fundamental
theorems of asset pricing apply to this market, that is the market
is arbitrage-free and complete which means that all derivatives can
be replicated by taking a dynamic position in the underlying assets.

In classical finance, the distinction between the underlying asset
(for example a stock) and a derivative (for example an option
on the stock) is natural. This is not the case in football betting; there is
no such clear distinction between underlying and derivative assets because all bets are made on the scores,
and the score process itself is not a tradable asset. In order to be able to apply the
Fundamental Theorems of Asset Pricing we need to artificially introduce underlying assets
and define the model by postulating a price dynamics for these assets in the physical measure.
It is also desirable to chose underlying assets that have a simple enough
price dynamics so that developing the replicating portfolio becomes as straightforward as possible.
For these reasons, the two underlying assets of our choice are assets that at the end of the game pay out the number of goals scored
by the home and away teams, respectively. It is important to note that these assets are not traded in practice
and the choice therefore seems unnatural. However, these underlying assets can be statically replicated from
Arrow-Debreu securities that are referred to as Correct Score bets in football betting and are traded in practice.
Furthermore, towards the end of the Section \ref{sec:riskneutralpricing} we arrive at Proposition \ref{prop:replication-from-anything} which
states that any two linearly independent bets can be used as hedging instruments. Therefore the choice of the underlying assets is practically
irrelevant and only serves a technical purpose. This result is applied in Section
\ref{sec:nextgoal} where Next Goal bets are used as natural hedging instruments.


\subsection{Setup}
Let us consider a probability space $\left(\Omega,\mathcal{F},\mathbb{P}\right)$
that carries two independent Poisson processes $N_{t}^{1}$, $N_{t}^{2}$
with respective intensities $\mu_{1}$, $\mu_{2}$ and the filtration
$\left(\mathcal{F}_{t}\right)_{t\in\left[0,T\right]}$ generated by
these processes. Let time $t=0$ denote the beginning and $t=T$ the
end of the game. The Poisson processes represent the number of goals
scored by the teams, the superscript $1$ refers to the home and $2$
refers to the away team. This notation is used throughout, the distinction
between superscripts and exponents will always be clear from the context.
The probability measure $\mathbb{P}$ is the real-world or physical
probability measure. 

We assume that there exists a liquid market where three assets can
be traded continuously with no transaction costs or any restrictions
on short selling or borrowing. The first asset $B_{t}$ is a risk-free
bond that bears no interests, an assumption that is motivated by the
relatively short time frame of a football game. The second and third
assets $S_{t}^{1}$ and $S_{t}^{2}$ are such that their values at
the end of the game are equal to the number of goals scored by the
home and away teams, respectively.

\begin{defn}[\bf model]
The model is defined by the following price dynamics of the assets:
\begin{eqnarray}
B_{t} & = & 1\nonumber \\
S_{t}^{1} & = & N_{t}^{1}+\lambda_{1}\left(T-t\right)\label{eq:modeldef}\\
S_{t}^{2} & = & N_{t}^{2}+\lambda_{2}\left(T-t\right)\nonumber 
\end{eqnarray}
where $\lambda_{1}$ and $\lambda_{2}$ are known real constants.
\end{defn}

Essentially, the underlying asset prices are compensated Poisson processes, but the compensators $\lambda_1,\lambda_2$
are not necessarily equal to the intensities $\mu_1,\mu_2$ and therefore the prices are not necessarilty
martingales in the physical measure $\mathbb{P}$. This is similar to the Black-Scholes model where the stock's drift in the physical measure
is not necessarily equal to the risk-free rate.

We are now closely following \citet{harrison1981martingales} in defining the necessary concepts.

\subsection{Risk-neutral pricing of bets}\label{sec:riskneutralpricing}
\begin{defn}[\bf trading strategy]
A \textit{trading strategy} is an $\mathcal{F}_{t}$-predictable
vector process $\phi_{t}=\left(\phi_{t}^{0},\phi_{t}^{1},\phi_{t}^{2}\right)$
that satisfies $\int_{0}^{t}\left|\phi_{s}^{i}\right|ds<\infty$ for
$i\in\left\{ 0,1,2\right\} $. The associated \textit{value process}
is denoted by
\begin{equation}
V_{t}^{\phi}=\phi_{t}^{0}B_{t}+\phi_{t}^{1}S_{t}^{1}+\phi_{t}^{2}S_{t}^{2}.
\end{equation}
 The trading strategy is \textit{self-financing }if 
\begin{equation}
V_{t}^{\phi}=V_{0}^{\phi}+\int_{0}^{t}\phi_{s}^{1}dS_{s}^{1}+\int_{0}^{t}\phi_{s}^{2}dS_{s}^{2}.
\end{equation}
where $\int_{0}^{t}\phi_{s}^{i}dS_{s}^{i}$, $i\in\left\{ 1,2\right\} $
is a Lebesgue Stieltjes integral which is well defined according to
Proposition 2.3.2 on p17 of \citet{bremaud1981point}.
\end{defn}

\begin{defn}[\bf arbitrage-freeness] 
The model is\textit{ arbitrage-free}
if no self-financing trading strategy $\phi_{t}$ exist such that
$\mathbb{P}\left[V_{t}^{\phi}-V_{0}^{\phi}\ge0\right]=1$ and $\mathbb{P}\left[V_{t}^{\phi}-V_{0}^{\phi}>0\right]>0$.
\end{defn}

\begin{defn}[\bf bet] 
A \textit{bet} (also referred to as a \emph{contingent claim} or \emph{derivative}) is an $\mathcal{F}_{T}$-measurable
random variable $X_{T}$.
\end{defn}

In practical terms this means that the value of a bet is revealed
at the end of the game.

\begin{defn}[\bf completeness] 
The model is \textit{complete} if for every
bet $X_{T}$ there exists a self-financing trading strategy $\phi_{t}$
such that $X_{T}=V_{T}^{\phi}$. In this case we say that the bet
$X_{T}$ is \textit{replicated} by the trading strategy $\phi_{t}$.
\end{defn}

\begin{thm}[risk-neutral measure]
\label{prop:equivalentMartingaleMeasure} There
exists a probability measure $\mathbb{Q}$ referred to as the risk-neutral
equivalent martingale measure such that:
\begin{enumerate}
\item[(a)]
The asset processes $B_{t}$, $S_{t}^{1}$, $S_{t}^{2}$ are $\mathbb{Q}$-martingales.
\item[(b)]
The goal processes $N_{t}^{1}$ and $N_{t}^{2}$ in measure $\mathbb{Q}$
are standard Poisson processes with intensities $\lambda_{1}$ and
$\lambda_{2}$ respectively (which are in general different from the
$\mathbb{P}$-intensities of $\mu_{1}$ and $\mu_{2}$).
\item[(c)]
$\mathbb{Q}$ is an equivalent measure to $\mathbb{P}$, that
is the set of events having zero probability is the same for both
measures.
\item[(d)]
$\mathbb{Q}$ is unique. 
\end{enumerate}
\end{thm}

\begin{proof}
The proof relies on Girsanov's theorem for point processes (see Theorem
2 on p.165 and Theorem 3 on page 166 in \citet{bremaud1981point})
which states that $N_{t}^{1}$ and $N_{t}^{2}$ are Poisson processes
with intensities $\lambda_{1}$ and $\lambda_{2}$ under the probability
measure $\mathbb{Q}$ which is defined by the Radon-Nikodym-derivative
\begin{equation}
\frac{d\mathbb{Q}}{d\mathbb{P}}=L_{t},
\end{equation}
 where 
\begin{equation}
L_{t}=\prod_{i=1}^{2}\left(\frac{\lambda_{i}}{\mu_{i}}\right)^{N_{t}^{i}}\exp\left[\left(\mu_{i}-\lambda_{i}\right)t\right].
\end{equation}
Then uniqueness follows from Theorem 8 on p.64 in \citet{bremaud1981point}
which states that if two measures have the same set of intensities,
then the two measures must coincide. The Integration Theorem on p.27
of \citet{bremaud1981point} states that $N_{t}^{i}-\lambda_{i}t$
are $\mathbb{Q}$-martingales, therefore the assets $S_{t}^{i}$ are
also $\mathbb{Q}$-martingales for $i\in\left\{ 1,2\right\} $. Proposition
9.5 of \citet{tankov2004financial} claims that $\mathbb{P}$ and
$\mathbb{Q}$ are equivalent probability measures. The process of
the bond asset $B_{t}$ is a trivial martingale in every measure because
it's a deterministic constant which therefore doesn't depend on the
measure.
\end{proof}

\begin{rem}
Changing the measure of a Poisson process changes the intensity and
leaves the drift unchanged. This is in contrast with the case of a
Wiener process where change of measure changes the drift and leaves
the volatility unchanged.
\end{rem}

\begin{thm}
\label{prop:arbFreeComplete}(arbitrage-free) The model is arbitrage-free
and complete.
\end{thm}

\begin{proof}
This follows directly from the first and second fundamental theorems
of finance. To be more specific, arbitrage-freeness follows from theorem
1.1 of \citet{delbaen1994general} which states that the existence
of a risk-neutral measure implies a so-called condition ``no free
lunch with vanishing risk'' which implies arbitrage-freeness. Completeness
follows from theorem 3.36 of \citet{harrison1981martingales} which
states that the model is complete if the risk-neutral measure is unique.
Alternatively it also follows from theorem 3.35 which states that
the model is complete if the martingale representation theorem holds
for all martingales which is the case according to Theorem 17, p.76
of \citet{bremaud1981point}.
\end{proof}

\begin{cor}
\label{prop:value_eq_expectedvalue}The time-$t$ value of a bet is
equal to the risk-neutral expectation of it's value at the end of the
game, formally:
\begin{equation}
X_{t}=\mathbf{E}^{\mathbb{Q}}\left[X_{T}|\mathcal{F}_{t}\right].
\end{equation}
\end{cor}

\begin{proof}
This follows directly from Proposition 3.31 of \citet{harrison1981martingales}.
\end{proof}

\begin{cor}
\label{prop:value_selffinancingstrategy}The time-$t$ value of a
bet is also equal to the value of the associated self-financing trading
strategy $\phi_{t}$, formally:
\begin{equation}
X_{t}=V_{t}^{\phi}=V_{0}^{\phi}+\int_{0}^{t}\phi_{s}^{1}dS_{s}^{1}+\int_{0}^{t}\phi_{s}^{2}dS_{s}^{2}.\label{eq:betvalue_replication}
\end{equation}
\end{cor}

\begin{proof}
This follows directly from Proposition 3.32 of \citet{harrison1981martingales}.
\end{proof}

\begin{defn}[\bf linear independence]
\label{def-linearindependence} 
The
bets $Z_{T}^{1}$ and $Z_{T}^{2}$ are \textit{linearly independent}
if the self-financing trading strategy $\phi_{t}^{1}=\left(\phi_{t}^{10},\phi_{t}^{11},\phi_{t}^{12}\right)$
that replicates $Z_{T}^{1}$ is $\mathbb{P}$-almost surely linearly
independent from the self-financing trading strategy $\phi_{t}^{2}=\left(\phi_{t}^{20},\phi_{t}^{21},\phi_{t}^{22}\right)$
that replicates $Z_{T}^{2}$. Formally, at any time $t\in\left[0,T\right]$
and for any constants $c_{1},c_{2}\in\mathbb{R}$ 
\begin{equation}
c_{1}\phi_{t}^{1}\ne c_{2}\phi_{t}^{2}\;\;\;\mathbb{P}\,{\it a.s.}
\end{equation}
\end{defn}

\begin{prop}[replication]
\label{prop:replication-from-anything} Any bet $X_{T}$
can be replicated by taking a dynamic position in any two linearly
independent bets $Z_{T}^{1}$ and $Z_{T}^{2}$, formally:
\begin{equation}
X_{t}=X_{0}+\int_{0}^{t}\psi_{s}^{1}dZ_{s}^{1}+\int_{0}^{t}\psi_{s}^{2}dZ_{s}^{2},\label{eq:replication_from_anything}
\end{equation}
where the weights $\psi_{t}^{1},\psi_{t}^{2}$ are equal to the solution
of the following equation:
\begin{equation}
\left(\begin{array}{cc}
\phi_{t}^{11} & \phi_{t}^{12}\\
\phi_{t}^{21} & \phi_{t}^{22}
\end{array}\right)\left(\begin{array}{c}
\psi_{t}^{1}\\
\psi_{t}^{2}
\end{array}\right)=\left(\begin{array}{c}
\phi_{t}^{1}\\
\phi_{t}^{2}
\end{array}\right)\label{eq:replication_equation}
\end{equation}
where $\left(\phi_{t}^{11},\phi_{t}^{12}\right)$, $\left(\phi_{t}^{21},\phi_{t}^{22}\right)$
and $\left(\phi_{t}^{1},\phi_{t}^{2}\right)$ are the components of
the trading strategy that replicates $Z_{T}^{1}$, $Z_{T}^{2}$ and
$X_{T}$, respectively. The integral $\int_{0}^{t}\psi_{s}^{1}dZ_{s}^{1}$
is to be interpreted in the following sense:
\begin{equation}
\int_{0}^{t}\psi_{s}^{1}dZ_{s}^{1}=\int_{0}^{t}\psi_{s}^{1}\phi_{s}^{11}dS_{s}^{1}+\int_{0}^{t}\psi_{s}^{1}\phi_{s}^{12}dS_{s}^{2}
\end{equation}
and similarly for $\int_{0}^{t}\psi_{s}^{2}dZ_{s}^{2}$.
\end{prop}

\begin{proof}
Substituting $dZ_{t}^{1}=\phi_{t}^{11}dS_{t}^{1}+\phi_{t}^{21}dS_{t}^{2}$,
$dZ_{t}^{2}=\phi_{t}^{12}dS_{t}^{1}+\phi_{t}^{22}dS_{t}^{2}$ and
Equation \ref{eq:betvalue_replication} into Equation \ref{eq:replication_from_anything}
verifies the proposition.
\end{proof}

\subsection{European bets}
\begin{defn}[\bf European bet]
\label{def:european}A \textit{European bet}
is a bet with a value depending only on the final number of goals
$N_{T}^{1}$, $N_{T}^{2}$, that is one of the form
\begin{equation}
X_{T}=\Pi\left(N_{T}^{1},N_{T}^{2}\right)
\end{equation}
where $\Pi$ is a known scalar function $\mathbb{N}\times\mathbb{N}\rightarrow\mathbb{R}$ which is referred to as the \textit{payoff function}.
\end{defn}

\begin{example}
A typical example is a bet that pays out $1$ if the home team scores
more goals than the away team (home wins) and pays nothing otherwise,
that is $\Pi\left(N_{T}^{1},N_{T}^{2}\right)=\mathbf{1}\left(N_{T}^{1}>N_{T}^{2}\right)$
where the function $\mathbf{1}\left(A\right)$ takes the value of
1 if $A$ is true and zero otherwise. Another example is a bet that
pays out $1$ if the total number of goals is strictly higher than
2 and pays nothing otherwise, that is $\Pi\left(N_{T}^{1},N_{T}^{2}\right)=\mathbf{1}\left(N_{T}^{1}+N_{T}^{2}>2\right)$.
\end{example}

\begin{prop}[pricing formula]
\label{prop:european_closedform} 
The time-$t$ value of a European bet with payoff function $\Pi$ is given by the explicit formula
\begin{equation}
X_{t}=\sum_{n_{1}=N_{1}^{t}}^{\infty}\sum_{n_{2}=N_{2}^{t}}^{\infty}\Pi\left(n_{1},n_{2}\right)P\left(n_{1}-N_{t}^{1},\lambda_{1}\left(T-t\right)\right)P\left(n_{2}-N_{t}^{2},\lambda_{2}\left(T-t\right)\right),\label{eq:europeanformula-1}
\end{equation}
where $P\left(N,\Lambda\right)$ is the Poisson probability,
that is $P\left(N,\Lambda\right)=\frac{e^{-\Lambda}}{N!}\Lambda^{N}$
if $N\ge0$ and $P\left(N,\Lambda\right)=0$ otherwise.
\end{prop}

\begin{proof}
This follows directly form Proposition \ref{prop:value_eq_expectedvalue}
and Definition \ref{def:european}.
\end{proof}

As we have seen, the price of a European bet is a function of the
time $t$ and the number of goals $\left(N_{t}^{1},N_{t}^{2}\right)$
and intensities $\left(\lambda_{1},\lambda_{2}\right)$. Therefore,
from now on we will denote this function by $X_{t}=X_{t}\left(N_{t}^{1},N_{t}^{2}\right)$
or $X_{t}=X_{t}\left(t,N_{t}^{1},N_{t}^{2},\lambda_{1},\lambda_{2}\right)$,
depending on whether the context requires the explicit dependence
on intensities or not.

It is important to note that Arrow-Debreu bets do exist in in-play football betting and are referred to as Correct Score bets.

\begin{defn}[\bf Arrow-Debreu bets]\label{defn:ad} \textit{Arrow-Debreu bets}, also known as Correct Score bets are European bets
with a payoff function $\Pi_{AD \left(K_1, K_2\right)}$ equal to $1$ if the final
score $\left(N_{T}^{1},N_{T}^{2}\right)$ is equal to a specified result $\left(K_1,K_2\right)$
and $0$ otherwise:
\begin{equation}
\Pi_{AD \left(K_1, K_2\right)} = \mathbf{1} \left( N_T^1=K_1, N_T^2=K_2\right)
\end{equation}
\end{defn}

According to the following proposition, Arrow-Debreu bets can be used to statically replicate any European bet:

\begin{prop}[static replication]
The time-$t$ value of a European bet with payoff function $\Pi$ in terms of time-$t$ values of Arrow-Debreu bets is given by:
\begin{equation}
X_{t}=\sum_{K_1=N_{1}^{t}}^{\infty}\sum_{K_2=N_{2}^{t}}^{\infty}\Pi\left(K_1,K_2\right)X_{t,AD\left(K_1,K_2\right)},
\end{equation}
where $X_{t,AD\left(K_1,K_2\right)}$ denotes the time-$t$ value of an Arrow-Debreu bet that pays out if the final scores are equal
to $\left(K_1,K_2\right)$.
\end{prop}

\begin{proof}
This follows directly form Proposition \ref{prop:european_closedform} and Definition \ref{defn:ad}.
\end{proof}

Let us now define the partial derivatives of the bet values with respect to change in time and the number goals scored. These are required for hedging and serve
the same purpose as the \textit{greeks} in the Black-Scholes framework.

\begin{defn}[Greeks]
\label{def:deltas_def}
The \textit{greeks} are the
values of the following forward difference operators ($\delta_{1}$,
$\delta_{2}$) and partial derivative operator applied to the bet
value:
\begin{eqnarray}
\delta_{1}X_{t}\left(N_{t}^{1},N_{t}^{2}\right) & = & X_{t}\left(N_{t}^{1}+1,N_{t}^{2}\right)-X_{t}\left(N_{t}^{1},N_{t}^{2}\right)\label{eq:delta1}\\
\delta_{2}X_{t}\left(N_{t}^{1},N_{t}^{2}\right) & = & X_{t}\left(N_{t}^{1},N_{t}^{2}+1\right)-X_{t}\left(N_{t}^{1},N_{t}^{2}\right)\label{eq:delta2}\\
\partial_{t}X_{t}\left(N_{t}^{1},N_{t}^{2}\right) & = & \lim_{dt\rightarrow0}\frac{1}{dt}\left[X_{t+dt}\left(N_{t}^{1},N_{t}^{2}\right)-X_{t}\left(N_{t}^{1},N_{t}^{2}\right)\right]
\end{eqnarray}
\end{defn}

\begin{rem*}
The forward difference operators $\delta_{1}$, $\delta_{2}$ play
the role of Delta and the partial derivative operator $\partial_{t}$
plays the role of Theta in the Black-Scholes framework.
\end{rem*}

\begin{thm}[Kolmogorov forward equation]
\label{prop:PIDE} The value of a European
bet $X\left(t,N_{t}^{1},N_{t}^{2}\right)$ with a payoff function
$\Pi(N_{T}^{1},N_{T}^{2})$ satisfies the following Feynman-Kac representation
on the time interval $t\in\left[0,T\right]$ which is also known as
the Kolmogorov forward equation: 
\begin{eqnarray}
\partial_{t}X\left(t,N_{t}^{1},N_{t}^{2}\right) & = & -\lambda_{1}\delta_{1}X\left(t,N_{t}^{1},N_{t}^{2}\right)-\lambda_{2}\delta_{2}X\left(t,N_{t}^{1},N_{t}^{2}\right)\label{eq:PIDE}
\end{eqnarray}
with boundary condition:
\[
X_{T}\left(T,N_{T}^{1},N_{T}^{2}\right)=\Pi\left(N_{T}^{1},N_{T}^{2}\right).
\]
\end{thm}

\begin{proof}
The proposition can be easily verified using the closed form formula
from Proposition \ref{prop:european_closedform}. Furthermore, several
proofs are available in the literature, see for example Proposition
12.6 in \citet{tankov2004financial}, Theorem 6.2 in \citet{ross2006introduction}
or Equation 13 in \citet{feller1940integro}.
\end{proof}

\begin{rem}
Equation \ref{eq:PIDE} also has the consequence that any portfolio
of European bets that changes no value if either team scores a goal
(Delta-neutral) does not change value between goals either (Theta-neutral).
We note without a proof, that this holds for all bets in general.
\end{rem}

\begin{cor}\label{prop:dX_dLambda}
The value of a European bet $X\left(t,N_{t}^{1},N_{t}^{2},\lambda_{1},\lambda_{2}\right)$
satisfies the following: 
\begin{eqnarray}
\frac{\partial}{\partial\lambda_{i}}X_{t} & = & \left(T-t\right)\delta_{i}X_{t}\label{eq:dX_dLambda}
\end{eqnarray}
where $i\in\left\{ 1,2\right\} $.
\end{cor}

\begin{proof}
This follows directly from Proposition \ref{prop:european_closedform}.
\end{proof}

\begin{prop}[portfolio weights]\label{prop:tradingstrategy_eq_deltas}
The components
$\left(\phi_{t}^{1},\phi_{t}^{2}\right)$ of the trading strategy
that replicates a European bet $X_{T}$ are equal to the forward difference
operators $\left(\delta_{1},\delta_{2}\right)$ of the bet, formally:
\begin{eqnarray}
\phi_{t}^{1} & = & \delta_{1}X\left(t,N_{t}^{1},N_{t}^{2}\right)\\
\phi_{t}^{2} & = & \delta_{2}X\left(t,N_{t}^{1},N_{t}^{2}\right).
\end{eqnarray}
\end{prop}

\begin{proof}
Recall that according to Proposition \ref{prop:value_selffinancingstrategy},
the time-$t$ value of a bet is equal to $X_{t}=X_{0}+\sum_{i=1}^{2}\int_{0}^{t}\phi_{s}^{i}dS_{s}^{i}$,
which after substituting $dS_{t}^{i}=dN_{t}^{i}-\lambda_{i}dt$ becomes
\begin{eqnarray}
X_{t} & = & X_{0}+\int_{0}^{t}\left(\phi_{s}^{1}\lambda_{1}+\phi_{s}^{2}\lambda_{2}\right)ds\nonumber \\
 &  & +\sum_{k=0}^{N_{t}^{1}}\phi_{t_{k}^{1}}^{1}+\sum_{k=0}^{N_{t}^{2}}\phi_{t_{k}^{2}}^{2},\label{eq:proof_replication_def}
\end{eqnarray}
 where we used $\int_{0}^{t}\phi_{s}^{i}dN_{s}^{i}=\sum_{k=0}^{N_{t}^{i}}\phi_{t_{k}^{1}}^{i}$
where $0\le t_{k}^{i}\le t$ is the time of the $k.$th jump (goal)
of the process $N_{t}^{i}$ for $i\in\left\{ 1,2\right\} $.

On the other hand, using Ito's formula for jump processes (Proposition
8.15, \citet{tankov2004financial}), which applies because the closed
form formula in Proposition \ref{prop:european_closedform} is infinitely
differentiable, the value of a European bet is equal to 
\begin{eqnarray}
X_{t} & = & X_{0}+\int_{0}^{t}\partial_{s}X\left(s,N_{s}^{1},N_{s}^{2}\right)ds\nonumber \\
 &  & +\sum_{k=0}^{N_{t}^{1}}\delta_{1}X\left(t_{k}^{1},N_{t_{k}^{1}-}^{1},N_{t_{k}^{1}-}^{2}\right)+\sum_{k=0}^{N_{t}^{2}}\delta_{2}X\left(t_{k}^{2},N_{t_{k}^{2}-}^{1},N_{t_{k}^{2}-}^{2}\right),\label{eq:proof_replication_ito}
\end{eqnarray}
where $t_{k}^{i}-$ refers to the fact that the value of the processes
is to be taken before the jump.

Because the equality between Equations \ref{eq:proof_replication_def}
and \ref{eq:proof_replication_ito} hold for every possible jump times,
the terms behind the sums are equal which proves the proposition.
\end{proof}

\section{Model Calibration}\label{sub:Calibration}
In this section we discuss how to calibrate the model parameters to historical
market prices. We demonstrate that a unique equivalent martingale measure
$\mathbb{Q}$ exists, that is, a set of intensities $\lambda_{1},\lambda_{2}$
exist that are consistent with the prices of all bets observed on the market (see
Propositions \ref{prop:equivalentMartingaleMeasure} and \ref{prop:arbFreeComplete}).

We apply a least squares approach in which we consider market prices
of a set of bets and find model intensities that deliver model prices
for these bets that are as close as possible to the market prices.
Specifically, we minimize the sum of the square of the weighted differences
between the model and market mid prices as a function of model intensities,
using market bid-ask spreads as weights. The reason for choosing a
bid-ask spread weighting is that we would like to take into account
bets with a lower bid-ask spread with a higher weight because the
price of these bets is assumed to be more certain. Formally, we minimize the following
expression:

\begin{equation}
R\left(\lambda_{t}^{1},\lambda_{t}^{2}\right)=\sqrt{\frac{1}{n}\sum_{i=1}^{n}\left[\frac{X_{t}^{i,{\it MID}}-X_{t}^{i}\left(\lambda_{t}^{1},\lambda_{t}^{2},N_{t}^{1},N_{t}^{2}\right)}{\frac{1}{2}\left(X_{t}^{i,{\it SELL}}-X_{t}^{i,{\it BUY}}\right)}\right]^{2}},\label{eq:calibration}
\end{equation}
where $n$ is the total number of bets used, $X_{t}^{i,{\it BUY}}$
and $X_{t}^{i,{\it SELL}}$ are the best market buy and sell quotes
of the $i.$th type of bet at time $t$, $X_{t}^{i,{\it MID}}$ is
the market mid price which is the average of the best buy and sell
quotes, $X_{t}^{i}\left(N_{t}^{1},N_{t}^{2},\lambda_{t}^{1},\lambda_{t}^{2}\right)$
is the model price of the $i$.th bet at time $t$, given the current
number of goals $N_{t}^{1},N_{t}^{2}$ and model intensity parameters
$\lambda_{t}^{1},\lambda_{t}^{2}$, see Proposition \ref{prop:european_closedform}.
This minimization procedure is referred to as model calibration.

Calibration has been performed using a time step of 1 minute during the
game, independently at each time step. We used the three
most liquid groups of bets which in our case were Match Odds, Over
/ Under and Correct Score with a total of 31 bet types in these three
categories. Appendix \ref{sec:Valuation-of-Bets} describes these
bet types in detail.

The continuous lines in Figures \ref{fig:Match-Odds-values.} and
\ref{fig:Over-Under-values.} show the calibrated model prices while
the dashed lines are the market buy and sell offers. It can be seen
that the calibrated values are close to the market quotes, although
they are not always within the bid-ask spread. As the measures of
the goodness of the fit we use the optimal value of the cost function
of Equation \ref{eq:calibration}, which is the average distance of
the calibrated values from the market mid prices in units of bid-ask
spread, the calibration error is shown in Figure \ref{fig:residual}. We performed
calibration for multiple games of the Euro 2012 Championship, the
time average of the calibration errors for each game is shown in Table
\ref{tab:Calibration-errors}. The mean and standard deviation of
the calibration errors across games is $1.57\pm0.27$ which is to
be interpreted in units of bid-ask spread because of the weighting
of the error function in Equation \ref{eq:calibration}. This means,
that on average, the calibrated values are outside of the bid-ask
spread, but not significantly. Given that a model of only 2 parameters
has been calibrated to a total of 31 independent market quotes, this
is a reasonably good result.

Finally, the implied intensities, along with the estimated uncertainties
of the calibration using the bid-ask spreads are shown in Figure \ref{fig:Calibrated-intensities}.
Contrary to our initial assumption of constant intensities, the actual intensities
fluctuate over time and there also seems to be an increasing trend in the implied goal intensities of both teams.

In order to better understand the nature of the implied intensity
process, we estimated the drift and volatility of the log total intensity,
that is we assumed the following:
\begin{equation}
d\ln\left(\lambda_{t}^{1}+\lambda_{t}^{2}\right)=\mu dt+\sigma dW_{t}
\end{equation}
where $\mu$ and $\sigma$ are the drift and volatility of the process.
Table \ref{tab:ModelparamDriftVol} shows the results of the estimation
for multiple games. The mean and standard deviation of the drift terms
are $\mu=0.55\pm0.16\;1/90{\it min}$ while the mean and standard
deviation of the volatility terms are $\sigma=0.51\pm0.19\;1/\sqrt{90{\it min}}$.
The fact that implied goal intensities are increasing during the game
is consistent with findings of \citet{dixon1998birth} who found gradual
increase of scoring rates by analysing goal times of 4012 matches
between 1993 and 1996.

\begin{figure}[t]
\begin{centering}
\includegraphics[width=0.55\paperwidth]{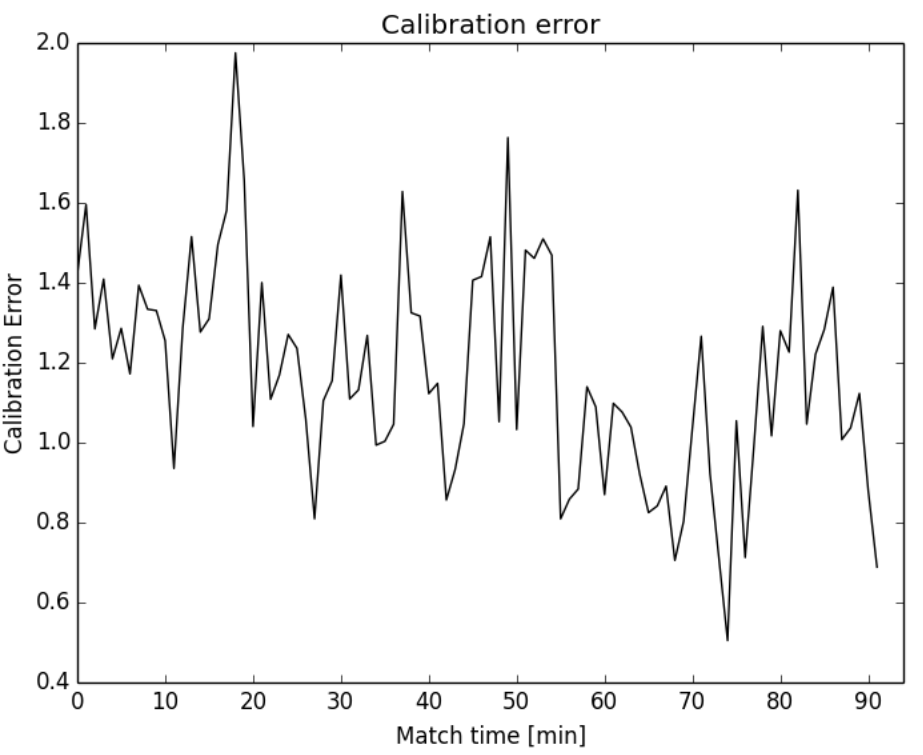}
\par\end{centering}
\protect\caption{\label{fig:residual}Calibration error during the game. Calibration
error is defined as the average distance of all 31 calibrated bet
values from the market mid prices in units of bid-ask spread. A formal
definition is given by Equation \ref{eq:calibration}. Note that the
calibration error for this particular game is usually between 1 and
2 bid-ask spreads which is a reasonably good result, given that the
model has only 2 free parameters to explain all 31 bet values.}
\end{figure}

\begin{figure}[t]
\begin{centering}
\includegraphics[width=0.55\paperwidth]{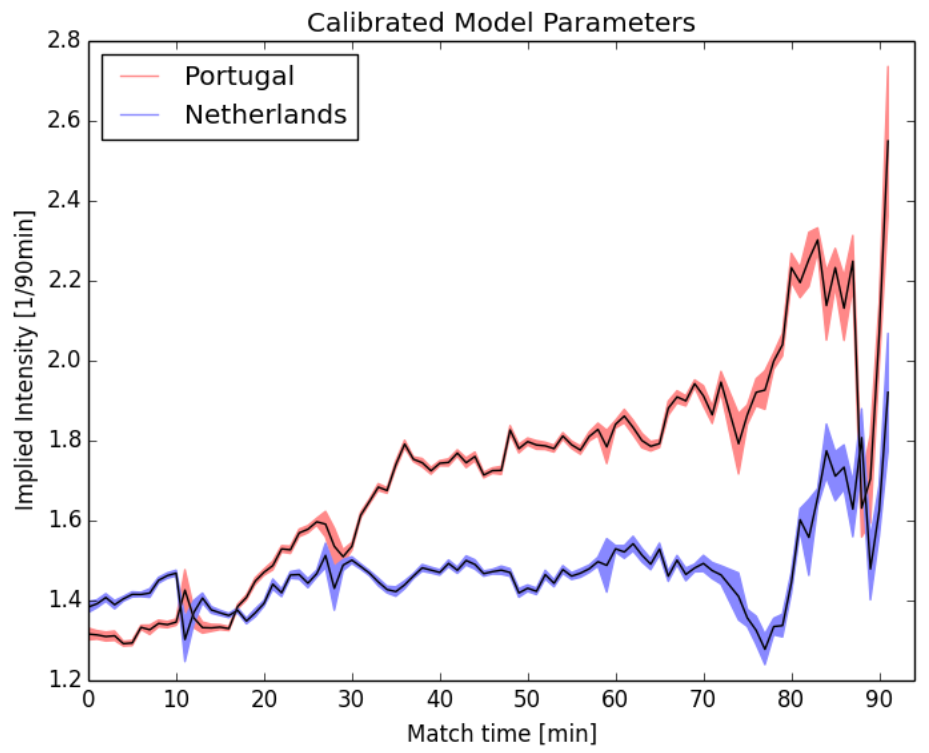}
\par\end{centering}
\protect\caption{\label{fig:Calibrated-intensities}Calibrated model parameters, also
referred to as implied intensities during the game. Formally, this
is equal to the minimizer $\lambda_{t}^{1},\lambda_{t}^{2}$ of Equation
\ref{eq:calibration}. The bands show the parameter uncertainties
estimated from the bid-ask spreads of the market values of the bets.
Note that the intensities appear to have an increasing trend and also
fluctuate over time.}
\end{figure}

\begin{table}[t]
\centering \begin{tabular}{|l|c|} \hline \textbf{Game} & \textbf{Calibration Error}\tabularnewline\hline Denmark v Germany & 1.65\tabularnewline\hline Portugal v Netherlands & 1.18\tabularnewline\hline Spain v Italy & 2.21\tabularnewline\hline Sweden v England & 1.58\tabularnewline\hline Italy v Croatia & 1.45\tabularnewline\hline Germany v Italy & 1.50\tabularnewline\hline Germany v Greece & 1.34\tabularnewline\hline Netherlands v Germany & 1.78\tabularnewline\hline Spain v Rep of Ireland & 1.64\tabularnewline\hline Spain v France & 1.40\tabularnewline\hhline{|=|=|} Average & 1.57\tabularnewline\hline Standard deviation & 0.27\tabularnewline\hline \end{tabular}  
\protect\caption{\label{tab:Calibration-errors}Average calibration errors in units of bid-ask spread as shown
in Figure \ref{fig:residual} have been calculated for multiple games
of the UEFA Euro 2012 Championship and are shown in this table. Note
that the mean of the averages is just 1.57 bid-ask spreads with a standard deviation
of 0.27 which shows that the model fit is reasonably good for
the games analysed.}
\end{table}

\begin{table}[t]
\centering \begin{tabular}{|l|c|c|} \hline \textbf{Game} & \textbf{Drift [$1/90\it{min}$]} & \textbf{Vol [$1/\sqrt{90\it{min}}$]}\tabularnewline\hline Denmark v Germany & 0.36 & 0.28\tabularnewline\hline Portugal v Netherlands & 0.49 & 0.44\tabularnewline\hline Spain v Italy & 0.60 & 0.76\tabularnewline\hline Sweden v England & 0.58 & 0.59\tabularnewline\hline Italy v Croatia & 0.82 & 0.60\tabularnewline\hline Germany v Italy & 0.76 & 0.39\tabularnewline\hline Germany v Greece & 0.65 & 0.66\tabularnewline\hline Netherlands v Germany & 0.43 & 0.32\tabularnewline\hline Spain v Rep of Ireland & 0.32 & 0.78\tabularnewline\hline Spain v France & 0.48 & 0.25\tabularnewline\hhline{|=|=|=|} Average & 0.55 & 0.51\tabularnewline\hline Standard deviation & 0.16 & 0.19\tabularnewline\hline \end{tabular} 
\protect\caption{\label{tab:ModelparamDriftVol}Average drift and volatility of total
log-intensities estimated for multiple games of the UEFA Euro 2012
Championship. Note that the drift term is positive for all games which
is consistent with the empirical observation of increasing goal frequencies
as the game progresses.}
\end{table}

\section{Hedging with Next Goal bets}\label{sec:nextgoal}

\begin{figure}[t]
\begin{centering}
\includegraphics[width=0.55\paperwidth]{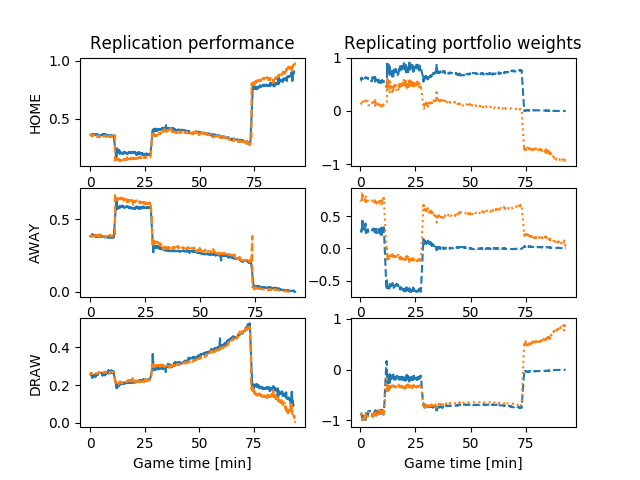}
\par\end{centering}
\protect\caption{\label{fig:ReplicationPerformance}Replicating the Match Odds home, away and draw contracts
using Next Goal home and away contracts as hedging instruments. The left column shows the replication performance
with the dashed line showing the value of the original Match Odds contracts and the continuous line showing the value of
the replicating portfolio. The right column shows the weights of the replicating portfolio with the dashed line
showing the weight of the Next Goal home contract and the dotted line showing the weight of the Next Goal away contract.}
\end{figure}

In this section we demonstrate market completeness and we show that
Next Goal bets are natural hedging instruments that can be used to
dynamically replicate and hedge other bets.

Recall that according to Proposition \ref{prop:replication-from-anything}
any European bet $X_{t}$ can be replicated by dynamically trading
in two linearly independent instruments $Z_{t}^{1}$ and $Z_{t}^{2}$:
\begin{equation}
X_{t}=X_{0}+\int_{0}^{t}\psi_{s}^{1}dZ_{s}^{1}+\int_{0}^{t}\psi_{s}^{2}dZ_{s}^{2}
\end{equation}
where the portfolio weights $\psi_{t}^{1},\psi_{t}^{2}$ are equal
to the solution of the equation
\begin{equation}
\left(\begin{array}{cc}
\delta_{1}Z_{t}^{1} & \delta_{1}Z_{t}^{2}\\
\delta_{2}Z_{t}^{1} & \delta_{2}Z_{t}^{2}
\end{array}\right)\left(\begin{array}{c}
\psi_{t}^{1}\\
\psi_{t}^{2}
\end{array}\right)=\left(\begin{array}{c}
\delta_{1}X_{t}\\
\delta_{2}X_{t}
\end{array}\right),\label{eq:replication_equation-1}
\end{equation}
where the values of the finite difference operators $\delta$ (Definition
\ref{def:deltas_def}) can be computed using Proposition \ref{prop:european_closedform}
using the calibrated model intensities. Equation \ref{eq:replication_equation-1}
tells us that the change in the replicating portfolio must match the
change of the bet value $X_{t}$ in case either team scores a goal.
This approach is analogous to delta hedging in the Black Scholes framework.

The two bets that we use as replicating instruments are the Next Goal home
and the Next Goal away bets. These bets settle during the game in a way such that
when the home team scores a goal the price of the Next Goal home bet jumps to 1 and the price of the Next Goal away bet jumps
to zero and vice versa for the away team. After the goal the bets reset and trade again at their regular market price. The values of the bets are:

\begin{eqnarray}
Z^{\it{NG_1}}_{t} & = & \frac{\lambda_{1}}{\lambda_{1}+\lambda_{2}}\left[1-e^{-\left(\lambda_{1}+\lambda_{2}\right)\left(T-t\right)}\right] \\
Z^{\it{NG_2}}_{t} & = & \frac{\lambda_{2}}{\lambda_{1}+\lambda_{2}}\left[1-e^{-\left(\lambda_{1}+\lambda_{2}\right)\left(T-t\right)}\right].
\end{eqnarray}

The matrix of deltas, that is the changes of contract values in case of a goal as defined in \ref{def:deltas_def} are:

\begin{equation}
\left(\begin{array}{cc}
\delta_{1}Z_{t}^{NG_1} & \delta_{1}Z_{t}^{NG_2}\\
\delta_{2}Z_{t}^{NG_1} & \delta_{2}Z_{t}^{NG_2}
\end{array}\right) =
\left(\begin{array}{cc}
1 - Z_{t}^{NG_1} & - Z_{t}^{NG_2}\\
- Z_{t}^{NG_1} & 1 - Z_{t}^{NG_2}
\end{array}\right)
\end{equation}

The reason for choosing Next Goal bets as hedging instruments is that these bets are linearly independent (see Definition \ref{def-linearindependence}), that is the delta matrix is non-singular even if there is a
large goal difference between the two teams. Note that this is an advantage compared to using the Match Odds bets as hedging instruments: in case one team leads by several goals,
it is almost certain that the team will win. In that case the value of the Match Odds bets goes close to 1 for the given team and 0 for the other team. An additional goal does not change the values
significantly, therefore the delta matrix becomes singular and the bets are not suitable for hedging because the portfolio weights go to infinity.
This is never the case with Next Goal bets which can therefore be used as natural hedging instruments.

We used the Portugal vs. Netherlands game from Section \ref{sec:examplegame} to replicate the values of the three Match Odds bets, using the Next Goal bets as
hedging instruments. Figure \ref{fig:ReplicationPerformance} shows the values of the original Match Odds bets along with the values of the replicating portfolios
(left column) and the replicating portfolio weights (right column).

Figure \ref{fig:ReplicationErrorNoGoals} shows the jumps of contract
values against the jumps of replicating portfolio values at times when a goal
was scored. This figure contains several different types of bets, that is not only Match Odds bets,
but also Over/Under and Correct Score bets. The figure also contains all 3 goals scored during the
Portugal vs. Netherlands game. It can be seen that the jumps of the original
contract values are in line with the jumps of the replicating portfolio
values with a correlation of 89\%. Table \ref{tab:ReplicationCorrelation}
shows these correlations for multiple games of the UEFA Euro 2012 Championship. It can be seen that the
correlations are reasonably high for all games with an average of
80\% and a standard deviation of 19\%.

\begin{figure}[t]
\begin{centering}
\includegraphics[width=0.55\paperwidth]{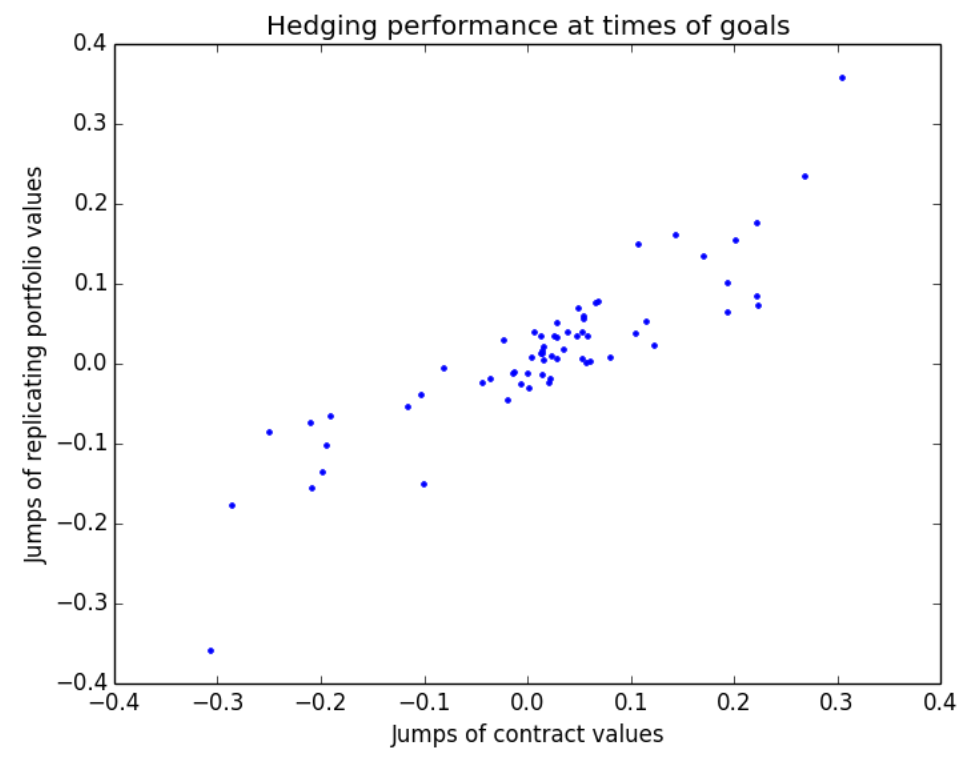}
\par\end{centering}
\protect\caption{\label{fig:ReplicationErrorNoGoals}Jumps of actual contract values
(horizontal axis) versus jumps of replicating
portfolio values (vertical axis) at times of goals scored during the Portugal vs. Netherlands game. The changes are
computed between the last traded price before a goal and the first
traded price after a goal, for all goals. The figure contains Match Odds, Over/Under and Correct Score bets.
Next Goal home and away bets were used as hedging instruments to build the replicating portfolios.
Note that the value changes of the replicating portfolios corresponds reasonably well to the value
changes of the original contracts with a correlation of 89\%.}
\end{figure}

\begin{table}[t]
\centering \begin{tabular}{|l|c|} \hline \textbf{Game} & \textbf{Correlation}\tabularnewline\hline Denmark vs. Germany & 79\%\tabularnewline\hline Portugal vs. Netherlands & 89\%\tabularnewline\hline Spain vs. Italy & 97\%\tabularnewline\hline Italy vs. Croatia & 47\%\tabularnewline\hline Spain vs. France & 86\%\tabularnewline\hline Germany vs. Italy & 99\%\tabularnewline\hline Germany vs. Greece & 60\%\tabularnewline\hline Netherlands vs. Germany & 93\%\tabularnewline\hline Spain vs. Rep of Ireland & 98\%\tabularnewline\hline Sweden vs. England & 50\%\tabularnewline\hhline{=|=} Average & 80\%\tabularnewline\hline Standard deviation & 19\%\tabularnewline\hline \end{tabular}
\protect\caption{\label{tab:ReplicationCorrelation}Correlation between the jumps of
bet values and jumps of replicating portfolios at times of goals for
all bets of a game.}
\end{table}

\section{Conclusions}\label{sec:conclusion}
In this paper we have shown that the Fundamental Theorems of Asset
Pricing apply to the market of in-play football bets if the scores
are assumed to follow independent Poisson processes of constant intensities.
We developed general formulae for pricing and replication. We have
shown that the model of only 2 parameters calibrates to 31 different
bets with an error of less than 2 bid-ask spreads. Furthermore, we
have shown that the model can also be used for replication and hedging.
Overall we obtained good agreement between actual contract values
and the values of the corresponding replicating portfolios, however
we point out that hedging errors can sometimes be significant due
to the fact the implied intensities are in practice not constant.


\section*{Funding}
PD and TA acknowledge support of the Economic and Social Research
Council (ESRC) in funding the Systemic Risk Centre (ES/K002309/1).
\bibliographystyle{rAMF}
\bibliography{references}

\appendix
\section{\label{sec:Valuation-of-Bets} Valuation formulae}
This section summarizes a list of analytical formulae for the values
of some of the most common in-play football bets. In the first sub-section
we consider European bets, while the second sub-section contains non-European
bets.

\subsection{European Bets}
The value of a European bet at the end of the game only depends on
the final scores. The formulae of this section follow directly from
Proposition \ref{prop:european_closedform}. Table \ref{tab:const_valudationformulae}
summarizes the payoff functions and the valuation formulae for some of the most common
types of European bets.

\begin{table}[t]
\centering
\begin{tabular}{|l|c|c|}
\hline 
\textbf{Bet type} & \textbf{Payoff $\Pi\left(N_{T}^{1},N_{T}^{2}\right)$} & \textbf{Value $X_{t}\left(N_{t}^{1},N_{t}^{2},\lambda_{1},\lambda_{2}\right)$}\tabularnewline
\hline 
\hline 
Match Odds Home & $\mathbf{1}\left(N_{T}^{1}>N_{T}^{2}\right)$ & $\sum_{k_{1}>k_{2}}\prod_{i=1}^{2}P\left(k_{i}-N_{t}^{i},\Lambda_{i}\right)$\tabularnewline
\hline 
Match Odds Away & $\mathbf{1}\left(N_{T}^{1}<N_{T}^{2}\right)$ & $\sum_{k_{1}<k_{2}}\prod_{i=1}^{2}P\left(k_{i}-N_{t}^{i},\Lambda_{i}\right)$\tabularnewline
\hline 
Match Odds Draw & $\mathbf{1}\left(N_{T}^{1}=N_{T}^{2}\right)$ & $\sum_{k_{1}=k_{2}}\prod_{i=1}^{2}P\left(k_{i}-N_{t}^{i},\Lambda_{i}\right)$\tabularnewline
\hline 
Arrow-Debreu $K_{1}$,$K_{2}$ & $\mathbf{1}\left(N_{T}^{1}=K_{1},N_{T}^{2}=K_{2}\right)$ & $\prod_{i=1}^{2}P\left(K_{i}-N_{t}^{i},\Lambda_{i}\right)$\tabularnewline
\hline 
Over $K$ & $\mathbf{1}\left(N_{T}^{1}+N_{T}^{2}>K\right)$ & $\sum_{k=K+1}^{\infty}P\left(k-N_{t}^{1}-N_{t}^{2},\left(\Lambda_{1}+\Lambda_{2}\right)\right)$\tabularnewline
\hline 
Under $K$ & $\mathbf{1}\left(N_{T}^{1}+N_{T}^{2}<K\right)$ & $\sum_{k=0}^{K-1}P\left(k-N_{t}^{1}-N_{t}^{2},\left(\Lambda_{1}+\Lambda_{2}\right)\right)$\tabularnewline
\hline 
Odd & $\mathbf{1}\left(N_{T}^{1}+N_{T}^{2}=1\mod2\right)$ & $\exp\left[-\left(\Lambda_{1}+\Lambda_{2}\right)\right]\cosh\left[\left(\Lambda_{1}+\Lambda_{2}\right)\right]$\tabularnewline
\hline 
Even & $\mathbf{1}\left(N_{T}^{1}+N_{T}^{2}=0\mod2\right)$ & $\exp\left[-\left(\Lambda_{1}+\Lambda_{2}\right)\right]\sinh\left[\left(\Lambda_{1}+\Lambda_{2}\right)\right]$\tabularnewline
\hline 
Winning Margin $K$ & $\mathbf{1}\left(N_{T}^{1}-N_{T}^{2}=K\right)$ & $\begin{array}{c}
\exp\left[-\left(\Lambda_{1}+\Lambda_{2}\right)\right]\left(\frac{\Lambda_{1}}{\Lambda_{2}}\right)^{\frac{K-N_{t}^{1}+N_{t}^{2}}{2}}\\
\cdot B_{\left|K-N_{t}^{1}+N_{t}^{2}\right|}\left(2\sqrt{\Lambda_{1}\Lambda_{2}}\right)
\end{array}$\tabularnewline
\hline 
\end{tabular}

\protect\caption{\label{tab:const_valudationformulae}Valuation formulae for some of the most common types of in-play football bets.
$\Pi\left(N_{T}^{1},N_{T}^{2}\right)$ denotes the payoff function,
that is the value of the European bet at the end of the game. $P\left(k,\Lambda\right)$
denotes the Poisson distribution, that is $P\left(k,\Lambda\right)=\frac{1}{k!}e^{-\Lambda}\Lambda^{k}$
and $\Lambda_{i}=\lambda_{i}\left(T-t\right)$ with $i\in\left\{ 1,2\right\}$ for the home and the away team, respectively.}

\end{table}

Match Odds Home, Away and Draw bets pay out depending on the final result of the game.
The Arrow-Debreu $K_{1}$,$K_{2}$ bets pay out if the final scores
are equal to $K_{1}$,$K_{2}$. Over $K$ and Under $K$ bets pay out if the total
number of goals is over or under $K$. Odd and Even bets pay out if the
total number of goals is an odd or an even number.

The Winning Margin $K$ bet wins if the difference between the home
and away scores is equal to $K$. The value of this bet follows the
Skellam distribution, $B_{k}\left(z\right)$ denotes the modified Bessel
function of the first kind.

\subsection{Non-European Bets}
Bets in this category have a value at the end of the game that depends
not only on the final score, but also on the score before the end
of the game or the order of scores. We consider two popular bets in
this category: Next Goal and Half Time / Full Time bets. Valuation
of these bets follows from Corollary \ref{prop:value_eq_expectedvalue}.

\subsubsection{Next Goal}
The Next Goal Home bet wins if the home team scores the next goal.
The value of this bet is

\begin{eqnarray}
X_{t} & = & \frac{\lambda_{1}}{\lambda_{1}+\lambda_{2}}\left[1-e^{-\left(\lambda_{1}+\lambda_{2}\right)\left(T-t\right)}\right].
\end{eqnarray}

Similarly, the value of the Next Goal Away bet is equal to
\begin{eqnarray}
X_{t} & = & \frac{\lambda_{2}}{\lambda_{1}+\lambda_{2}}\left[1-e^{-\left(\lambda_{1}+\lambda_{2}\right)\left(T-t\right)}\right].
\end{eqnarray}

\subsubsection{Half Time / Full Time}
Half Time / Full Time bets win if the half time and the full time
is won by the predicted team or is a draw. Given that there are 3
outcomes in each halves, there are 9 bets in this category. For example,
the value of the Half Time Home / Full Time Draw bet before the end
of the first half is equal to:

\begin{eqnarray}
X_{t} & =\sum_{k_{1}>k_{2}}\sum_{l_{1}=l_{2}} & P\left(k_{1}-N_{t}^{1},\lambda_{1}\left(T_{\frac{1}{2}}-t\right)\right)P\left(k_{2}-N_{t}^{2},\lambda_{2}\left(T_{\frac{1}{2}}-t\right)\right)\nonumber \\
 &  & \times P\left(l_{1}-k_{1},\lambda_{1}\left(T-T_{\frac{1}{2}}\right)\right)P\left(l_{2}-k_{2},\lambda_{2}\left(T-T_{\frac{1}{2}}\right)\right).
\end{eqnarray}

In the second half, this bet either becomes worthless if the first
half was not won by the home team or otherwise becomes equal to the
Draw bet.

\end{document}